\documentclass[12pt]{article}

\usepackage{amsmath,amssymb,amsthm}
\newtheorem{theorem}{Theorem}
\newtheorem{corollary}[theorem]{Corollary}
\newtheorem{definition}[theorem]{Definition}
\newtheorem{lemma}[theorem]{Lemma}
\newtheorem{proposition}[theorem]{Proposition}

\begin{document}
\title{Imperfect cloning operations \\ in algebraic quantum theory}
\date{}
\author{Yuichiro Kitajima}
\maketitle

\begin{abstract}
No-cloning theorem says that there is no unitary operation that makes perfect clones of non-orthogonal quantum states. The objective of the present paper is to examine whether an imperfect cloning operation exists or not in a C*-algebraic framework. We define a universal $\epsilon$-imperfect cloning operation which tolerates a finite loss $\epsilon$ of fidelity in the cloned state, and show that an individual system's algebra of observables is abelian if and only if there is a universal $\epsilon$-imperfect cloning operation in the case where the loss of fidelity is less than $1/4$. Therefore in this case no universal $\epsilon$-imperfect cloning operation is possible in algebraic quantum theory.

\end{abstract}

\section{Introduction}
Dieks \cite{dieks1982communication} and Wootters and Zurek \cite{wootters1982single} showed that there is no unitary operation that makes clones of non-orthogonal quantum states (cf. \cite[p.532]{nielsen2010quantum}). 
It is called no-cloning theorem, and this property is one of the fundamental differences between classical and quantum information.
Clifton, Bub and Halvorson \cite{clifton2003characterizing} generalized the notion of cloning to C*-algebraic states, and showed that an individual system's algebra of observables is abelian if and only if there is a universal cloning operation. 

No-cloning theorem applied only to perfect cloning. If we allow imperfect cloning operations which are `good' according to fidelity, do they exist in quantum mechanics? Bu{\v{z}}ek and Hillery \cite{buvzek1996quantum} showed that there exists a universal imperfect cloning operation which tolerates a finite loss of fidelity in the cloned state. The essential point of this operation is that the original and the cloned states are entangled. The objective of the present paper is to examine the case where the original and the cloned states are not entangled and a finite loss of fidelity in the cloned state is tolerated, and it is shown that such a universal imperfect cloning operation does not exist in algebraic quantum theory.

This paper is organized as follows. We begin, in section \ref{aqft}, by laying out a C*-algebraic framework. After introducing the framework, we examine the relations between a fidelity and a transition probability in section \ref{fidelities}, which play an important role in no-cloning theorem. 
In section \ref{imperfect_cloning} we define a universal $\epsilon$-imperfect cloning operation, which tolerates a finite loss $\epsilon$ of fidelity in the cloned state (Definition \ref{imperfect_definition}). In Theorem \ref{imperfect_theorem} it is shown that any individual system's algebra of observables is abelian if there is a universal $\epsilon$-imperfect cloning operation in the case where the loss of fidelity is less than $1/4$. On the other hand, such an algebra is not abelian in algebraic quantum theory. Therefore in this case no universal $\epsilon$-imperfect cloning operation is possible in algebraic quantum theory.

\section{Preliminary}
\label{aqft}

In this section, we shall introduce a C*-algebraic framework in order to apply it to algebraic quantum field theory. Algebraic quantum field theory exists in two versions: the Haag-Araki theory which uses von Neumann algebras on a Hilbert space, and the Haag-Kastler theory which uses abstract C*-algebras (cf. \cite{horuzhy1990introduction}). Here we examine no-cloning theorem in the Haag-Kastler theory.

In this theory, each bounded open region $\mathcal{O}$ in the Minkowski space is associated with a unital C*-algebra $\mathfrak{A}(\mathcal{O})$. Such a C*-algebra is called a local algebra. The set theoretic union of all $\mathfrak{A}(\mathcal{O})$ is a normed *-algebra. Taking its completion we get a C*-algebra $\mathfrak{A}_0$. Thus each local algebra is contained in $\mathfrak{A}_0$.
The following assumptions are made in the Haag-Kastler theory (\cite{haag1964algebraic}, cf. \cite{horuzhy1990introduction}).
\begin{description}
\item[Microcausality] If $\mathcal{O}_1$ and $\mathcal{O}_2$ are space-like separated, then $X_1X_2-X_2X_1=0$ for any $X_1 \in \mathfrak{A}(\mathcal{O}_1)$ and $X_2 \in \mathfrak{A}(\mathcal{O}_2)$. 
\item[C*-independence] If $\mathcal{O}_1$ and $\mathcal{O}_2$ are space-like separated, for any states $\psi_1$ of $\mathfrak{A}(\mathcal{O}_1)$ and $\psi_2$ of $\mathfrak{A}(\mathcal{O}_2)$, there exists a state $\psi$ of $\mathfrak{A}_0$ such that $\psi|_{\mathfrak{A}(\mathcal{O}_1)}=\psi_1$ and $\psi|_{\mathfrak{A}(\mathcal{O}_2)}=\psi_2$.
\item[Relativistic covariance] Let $g=(\Lambda, a)$ denote a Poincar\'e transformation $x \in M \rightarrow \Lambda x + a \in M$, where $a \in M$ is the amount of space-time translation and $\Lambda$ is a Lorentz transformation. There exists a representation of the Poincar\'{e} transformation by the automorphisms $\alpha_{(a,\Lambda)}$ of $\mathfrak{A}_0$ such that
\[ \alpha_{(a, \Lambda)}(\mathfrak{A}(\mathcal{O}))=\mathfrak{A}(\Lambda^{-1}(\mathcal{O}-a)). \]
\end{description}


The following C*-algebraic framework can be applied to the Haag-Kastler theory (cf. \cite[Section 3.3]{clifton2003characterizing}).
Let $\mathfrak{A}_1$ and $\mathfrak{A}_2$ be unital C*-subalgebras of a unital C*-algebra $\mathfrak{A}$. Throughout this paper, we suppose that they satisfy the following conditions;
\begin{itemize}
\item $A_1A_2 - A_2A_1=0$ for any $A_1 \in \mathfrak{A}_1$ and $A_2 \in \mathfrak{A}_2$;
\item For any states $\psi_1$ of $\mathfrak{A}_1$ and $\psi_2$ of $\mathfrak{A}_2$, there exists a state $\psi$ of $\mathfrak{A}$ such that $\psi|_{\mathfrak{A}_1}=\psi_1$ and $\psi|_{\mathfrak{A}_2}=\psi_2$;
\item There is a *-isomorphism $\alpha$ of $\mathfrak{A}_1$ onto $\mathfrak{A}_2$. We say that states $\psi_1$ of $\mathfrak{A}_1$ and $\psi_2$ of $\mathfrak{A}_2$ are isomorphic when $\psi_1(A_1)=\psi_2(\alpha(A_1))$ for any $A_1 \in \mathfrak{A}_1$. 
\end{itemize}

Let define a map $\eta$ from the *-algebra generated by $\mathfrak{A}_1$ and $\mathfrak{A}_2$ onto the algebraic tensor product $\mathfrak{A}_1 \odot \mathfrak{A}_2$ by $\eta(A_1A_2)=A_1 \otimes A_2$ for all $A_1 \in \mathfrak{A}_1$ and $A_2 \in \mathfrak{A}_2$ and let $\mathfrak{A}_1 \otimes \mathfrak{A}_2$ be the completion of $\mathfrak{A}_1 \odot \mathfrak{A}_2$ under the injective C*-cross norm $\| \cdot \|_{\text{min}}$. This norm is given by
\[ \| A \|_{\text{min}}=\sup \| (\pi_1 \otimes \pi_2)(A) \| \]
for any $A \in \mathfrak{A}_1 \odot \mathfrak{A}_2$, where the supremum is taken over all representations of $\mathfrak{A}_1$ and $\mathfrak{A}_2$, respectively. 

Then $\eta$ is an algebraic isomorphism with respect to the injective C*-cross norm on $\mathfrak{A}_1 \odot \mathfrak{A}_2$, and can be extended to a continuous homomorphism $\bar{\eta}$ from $\mathfrak{A}_1 \vee \mathfrak{A}_2$ onto $\mathfrak{A}_1 \otimes \mathfrak{A}_2$, where $\mathfrak{A}_1 \vee \mathfrak{A}_2$ is a C*-algebra generated by $\mathfrak{A}_1$ and $\mathfrak{A}_2$ (\cite{roos1970independence} and \cite[Theorem IV.4.9]{takesaki2002theory}, cf. \cite[Theorem 1.3.25]{horuzhy1990introduction}). So we use $A_1 \otimes A_2$ to denote the product of $A_1 \in \mathfrak{A}_1$ and $A_2 \in \mathfrak{A}_2$.

For any state $\omega$ of $\mathfrak{A}_1 \otimes \mathfrak{A}_2$, let define the state $\bar{\eta}^*\omega$ of $\mathfrak{A}_1 \vee \mathfrak{A}_2$ by setting
\[ (\bar{\eta}^* \omega)(A)=\omega(\bar{\eta}(A)) \]
for any $A \in \mathfrak{A}_1 \vee \mathfrak{A}_2$. If $\psi_1$ and $\phi_2$ are states of  $\mathfrak{A}_1$ and $\mathfrak{A}_2$, respectively, $\bar{\eta}^*(\psi_1 \otimes \phi_2)$ is a product state of $\mathfrak{A}_1 \vee \mathfrak{A}_2$ with marginal states $\psi_1$ and $\phi_2$. Clifton, Bub and Halvorson \cite[Lemma 1]{clifton2003characterizing} showed that $\bar{\eta}^*(\psi_1 \otimes \phi_2)$ is the unique product state of $\mathfrak{A}_1 \vee \mathfrak{A}_2$ with these marginal states. When it will not cause confusion, we will use $\psi_1 \otimes \phi_2$ to denote the unique product state of $\mathfrak{A}_1 \vee \mathfrak{A}_2$ with marginals $\psi_1$ and $\phi_2$.

For representations $\pi_1$ of $\mathfrak{A}_1$ on a Hilbert space $\mathcal{H}_1$ and $\pi_2$ of $\mathfrak{A}_2$ on a Hilbert space $\mathcal{H}_2$, the representation $\pi_1 \otimes \pi_2$ of $\mathfrak{A}_1 \odot \mathfrak{A}_2$ on $\mathcal{H}_1 \otimes \mathcal{H}_2$ is uniquely extended to a representation of $\mathfrak{A}_1 \otimes \mathfrak{A}_2$, which is denoted again by $\pi_1 \otimes \pi_2$ \cite[p.208]{takesaki2002theory}. Thus $(\pi_1 \otimes \pi_2) \circ \bar{\eta}$ is a representation of $\mathfrak{A}_1 \vee \mathfrak{A}_2$. We will omit reference to $\bar{\eta}$ and use $\pi_1 \otimes \pi_2$ to denote $(\pi_1 \otimes \pi_2) \circ \bar{\eta}$ when it will not cause confusion.

This C*-algebraic framework can also be applied to a case of a finite-dimensional Hilbert space. Let $\mathbb{B}(\mathcal{H}_n)$ be the set of all operators on a finite-dimensional Hilbert space $\mathcal{H}_n$, and let $I$ be an identity operator on $\mathcal{H}_n$. Then $\mathbb{B}(\mathcal{H}_n) \otimes I$ and $I \otimes \mathbb{B}(\mathcal{H}_n)$ are mutually commuting and C*-independent. 
Let $\alpha$ be a mapping from $\mathbb{B}(\mathcal{H}_n) \otimes I$ to $I \otimes \mathbb{B}(\mathcal{H}_n)$ such that $\alpha(A \otimes I)=I \otimes A$ for any $A \in \mathbb{B}(\mathcal{H}_n)$. This is an isomorphism from $\mathbb{B}(\mathcal{H}_n) \otimes I$ onto $I \otimes \mathbb{B}(\mathcal{H}_n)$. Therefore we can apply this C*-algebraic framework to such a case.

\section{Fidelity}
\label{fidelities}
In this section, we examine a fidelity, which plays an important role in no-broadcasting theorem. The fidelity $F(\psi, \phi)$ is defined as follows (\cite[Definition 1.2]{bures1969extension}, cf. \cite[Section 9.2.2]{nielsen2010quantum}).

\begin{definition}
Suppose that  $\pi$ is a representation of $\mathfrak{A}$ on a Hilbert space $\mathcal{H}$. For each state $\psi$ of $\mathfrak{A}$ define by
\[ S(\pi, \psi)=\{ x \in \mathcal{H} | \psi(A) = \langle x, \pi(A)x \rangle \ \text{for all} \ A \in \mathfrak{A} \}. \]

Let $\psi$ and $\phi$ be states of $\mathfrak{A}$, let $\Pi(\mathfrak{A})$ be the set of all representations of $\mathfrak{A}$, and let $\pi$ be in $\Pi(\mathfrak{A})$. If either $S(\pi, \psi)$ or $S(\pi, \phi)$ is empty, define 
$D_{\pi}(\psi, \phi)=\sqrt{2}$ and 
$F_{\pi}(\psi, \phi)=0$; otherwise define
\[ D_{\pi}(\psi, \phi)= \inf \{ \| x-y \| | x \in S(\pi, \psi), y \in S(\pi, \phi) \}, \]
\[ F_{\pi}(\psi, \phi)= \sup \{ |\langle x, y \rangle | | x \in S(\pi, \psi), y \in S(\pi, \phi) \}. \]
Moreover we define as follows.
\[ D(\psi,\phi)=\inf \{ D_{\pi}(\psi, \phi) | \pi \in \Pi(\mathfrak{A}) \}, \]
\[ F(\psi,\phi)=\sup \{F_{\pi}(\psi, \phi) | \pi \in \Pi(\mathfrak{A}) \}. \]
\end{definition}

$D(\psi,\phi)$ can be written by $F(\psi, \phi)$ \cite[Lemma 1.4]{bures1969extension} and satisfies the triangle inequality \cite[Proposition 1.7]{bures1969extension};
\begin{equation}
\label{fidelity_distance}
D(\psi,\phi)^2=2-2F(\psi,\phi);
\end{equation}
\begin{equation}
\label{triangle}
D(\psi,\phi) \leq D(\psi, \omega) + D(\omega, \phi)
\end{equation}
for any states $\psi$, $\phi$ and $\omega$ of $\mathfrak{A}$.

Uhlmann \cite[Section 5]{uhlmann1976transition} showed that $F(\psi, \phi)=\text{tr}(D_{\psi}^{1/2}D_{\phi}D_{\psi}^{1/2})^{1/2}$ when $\mathfrak{A}$ is the algebra $\mathbb{B}(\mathcal{H})$ of all bounded operators on a Hilbert space $\mathcal{H}$, and $\psi$ and $\phi$ are states of $\mathbb{B}(\mathcal{H})$ which are written by density operators $D_{\psi}$ and $D_{\phi}$, respectively.
$\text{tr}(D_{\psi}^{1/2}D_{\phi}D_{\psi}^{1/2})^{1/2}$ is called a fidelity, which is a quantitative measure of similarity between two states (cf. \cite[Section 9.2.2]{nielsen2010quantum}). 

There are many works about properties of a fidelity. For example, Araki and Raggio 
\cite{araki1982remark,raggio1982comparison} examine them in a von Neumann algebraic framework, using the theory of von Neumann algebra in standard form, and Alberti and Uhlmann \cite{alberti1983note,alberti1983stochastic,uhlmann1976transition,uhlmann1985transition} examine them in a C*-algebraic framework. In this section, we state its properties which are needed for the proof of the main theorem.

Let $\psi_i$ and $\phi_i$ be states of $\mathfrak{A}_i$ for $i=1,2$. Suppose that $\psi_2$ and $\phi_2$ are isomorphic and $\psi_1$ and $\phi_1$, respectively, that is, there is a *-isomorphism $\alpha$ of $\mathfrak{A}_1$ onto $\mathfrak{A}_2$ such that $\psi_1(A_1)=\psi_2(\alpha(A_1))$ and $\phi_1(A_1)=\phi_2(\alpha(A_1))$ for any $A_1 \in \mathfrak{A}_1$. If $x_1 \in S(\pi, \psi_1)$ and $y_1 \in S(\pi, \phi_1)$, then $x_1 \in S(\pi \circ \alpha^{-1}, \psi_2)$ and $y_1 \in S(\pi \circ \alpha^{-1}, \phi_2)$. Thus $| \langle x_1, y_1 \rangle | \leq F(\psi_2, \phi_2)$. Taking the supremum over $x_1 \in S(\pi, \psi_1)$ and $y_1 \in S(\pi, \phi_1)$, $F(\psi_1,\phi_1) \leq F(\psi_2, \phi_2)$. Similarly $F(\psi_1,\phi_1) \geq F(\psi_2, \phi_2)$. Therefore
\begin{equation}
\label{e0.01}
F(\psi_1,\phi_1)=F(\phi_2, \phi_2).
\end{equation}

Although Bures \cite[Proposition 1.6]{bures1969extension} showed the following proposition in the case of W*-algebras, it also holds in the case of C*-algebras.

\begin{proposition}
\label{Bures_representation}
There exists a representation $\pi$ of $\mathfrak{A}$ on a Hilbert space $\mathcal{H}$ such that 
$F(\psi, \phi)=F_{\pi}(\psi, \phi)$ for any states $\psi$ and  $\phi$ of $\mathfrak{A}$.
\end{proposition}

\begin{definition}
We call the representation in Proposition \ref{Bures_representation} Bures representation of $\mathfrak{A}$.
\end{definition}


Next we introduce a transition probability. Roberts and Roepstorff \cite[Definition 4.7]{roberts1969some} proved the following proposition.

\begin{proposition}
\label{roberts}
Let $\mathfrak{B}$ be a C*-algebra on a Hilbert space $\mathcal{H}$ and let $\omega_x$ and $\omega_y$ be the pure states of $\mathfrak{B}$ induced by the unit vectors $x$ and $y$. Then
\[ | \langle x,y \rangle |^2=1-\frac{1}{4}\| \omega_x - \omega_y \|^2. \]
\end{proposition}

Based on this proposition they defined a transition probability $\psi \cdot \phi$ between pure states $\psi$ and $\phi$ of $\mathfrak{A}$ as follows \cite[Definition 4.7]{roberts1969some}.

\begin{definition}
Let $\psi$ and $\phi$ be pure states of $\mathfrak{A}$. Let define
\[ \psi \cdot \phi = 1-\frac{1}{4}\| \psi - \phi \|^2. \]
\end{definition}

In the following proposition, we examine the relation between $F(\psi, \phi)$ and $\psi \cdot \phi$ in a C*-algebraic framework. 

\begin{proposition}
\label{transition}
Let $\psi$ and $\phi$ be pure states of $\mathfrak{A}$. Then $\psi \cdot \phi = F(\psi,\phi)^2$.
\end{proposition}

In order to prove Proposition \ref{transition} and  Proposition \ref{nonabelian}, the following lemma is needed.

\begin{lemma}
\label{inner_product}
Let $\psi$ be a pure state of $\mathfrak{A}$, let $(\pi_{\psi}, \mathcal{H}_{\psi}, x_{\psi})$ be GNS representation induced by $\psi$, let $y$ be a unit vector in $\mathcal{H}_{\psi}$ and let $\phi$ be a state of $\mathfrak{A}$ such that $\phi(A)=\langle y, \pi(A)y \rangle$ for any $A \in \mathfrak{A}$. Then $\| \psi - \phi \| = 2(1-|\langle x_{\psi}, y \rangle |^2)^{1/2}$.
\end{lemma}

\begin{proof}
Let $P$ and $Q$ be projections whose ranges are subspaces generated by $x_{\psi}$ and $y$, respectively. Then
\begin{equation}
\label{transition01}
\psi(A)-\phi(A)=\text{tr}((P-Q)\pi_{\psi}(A))
\end{equation}
for any $A \in \mathfrak{A}$. 

If $P=Q$, then $\| \psi - \phi \| = 0 =2(1-|\langle x_{\psi}, y \rangle |^2)^{1/2}$. Suppose that $P \neq Q$. Then there are orthogonal projections $R_1$ and $R_2$ and $a_1,a_2 \in \mathbb{R}$ such that 
\begin{equation}
\label{transition02}
P-Q=a_1R_1-a_2R_2.
\end{equation}
By taking the trace of Equation (\ref{transition02}), we obtain $a_1=a_2$. By squaring Equation (\ref{transition02}) and taking the trace 
\begin{equation}
\label{transition03}
2-2|\langle x_{\psi},y \rangle |^2=2a_1^2.
\end{equation}
Thus $| \psi(A)-\phi(A) |=a_1 \text{tr}((R_1-R_2)\pi(A)) \leq 2a_1 \leq 2(1-| \langle x_{\psi},y \rangle |^2)^{1/2}$ for any $A \in \mathfrak{A}$ such that $\| A \| \leq 1$ by Equations (\ref{transition01}), (\ref{transition02}) and (\ref{transition03}) (cf. \cite[Section 3]{uhlmann1976transition}). Therefore
\begin{equation}
\label{transition04}
\| \psi - \phi \| \leq 2(1-| \langle x_{\psi},y \rangle |^2)^{1/2}.
\end{equation}

Let $V^0=R_1-R_1^{\perp}$ and let $z_1$ and $z_2$ be unit vectors which are in the range of $R_1$ and $R_2$, respectively. Then $V^0z_1=z_1$ and $V^0z_2=-z_2$. Since $V^0$ is a unitary operator, there is a unitary element $V \in \mathfrak{A}$ such that $\pi_{\psi}(V)z_1=z_1$ and $\pi_{\psi}(V)z_2=-z_2$ \cite[Theorem 10.2.1]{kadison1986fundamentals}. By Equations (\ref{transition01}), (\ref{transition02}) and (\ref{transition03}),
\begin{equation}
\label{transition05}
 \begin{split}
&\| \psi - \phi \| \geq |\psi(V)-\phi(V)|=|\text{tr}((P-Q)\pi_{\psi}(V))|=|\text{tr}((a_1R_1-a_2R_2)\pi_{\psi}(V))| \\
&=2a_1=2(1-|\langle x_{\psi}, y \rangle |^2)^{1/2} .
\end{split} 
\end{equation}
Equations (\ref{transition04}) and (\ref{transition05}) imply $\| \psi - \phi \|=2(1-|\langle x_{\psi}, y \rangle |^2)^{1/2}$.
 \end{proof}

In the present paper, if $\mathbb{A}$ is a set of operators acting on a Hilbert space $\mathcal{H}$, let $\mathbb{A}'$ represent its commutant, the set of all bounded operators on $\mathcal{H}$ which commute with all elements of $\mathbb{A}$.

\begin{proof}[Proof of Proposition \ref{transition}]
Let $(\pi_{\psi},x_{\psi},\mathcal{H}_{\psi})$ and $(\pi_{\phi},x_{\phi},\mathcal{H}_{\phi})$ be GNS representations induced by $\psi$ and $\phi$, respectively.

\begin{enumerate}
\item
Suppose that $\pi_{\psi}$ and $\pi_{\phi}$ are not unitarily equivalent. Then $\| \psi - \phi \|=2$ by \cite[Corollary 10.3.8]{kadison1986fundamentals}. Thus $\psi \cdot \phi=0$.

Let $\pi$ be Bures representation of $\mathfrak{A}$ on a Hilbert space $\mathcal{H}$. Let $x \in S(\pi, \psi)$ and $y \in S(\pi, \phi)$, let $E' \in \pi(\mathfrak{A})'$ and $F' \in \pi(\mathfrak{A})'$ be projections whose range are subspaces generated by $\{ \pi(\mathfrak{A})x \}$ and $\{ \pi(\mathfrak{A})y \}$, respectively, and let $\pi_{E'}$ and $\pi_{F'}$ be representations on $E'\mathcal{H}$ and $F'\mathcal{H}$ such that $\pi_{E'}(A)=E'\pi(A)E'$ and $\pi_{F'}(A)=F'\pi(A)F'$ for any $A \in \mathfrak{A}$, respectively. Then $\pi_{E'}$ and $\pi_{F'}$ are unitarily equivalent to $\pi_{\psi}$ and $\pi_{\phi}$, respectively \cite[Proposition 4.5.3]{kadison1983fundamentals}. Since $\pi_{\psi}$ and $\pi_{\phi}$ are unitarily inequivalent, so are $\pi_{E'}$ and $\pi_{F'}$. Thus $\pi_{E'}$ and $\pi_{F'}$ are disjoint \cite[Proposition 10.3.7]{kadison1986fundamentals}. In this case $E'F'=0$ \cite[Proposition 10.3.3]{kadison1986fundamentals}. Since $x \in E'\mathcal{H}$ and $y \in F'\mathcal{H}$, $\langle x, y \rangle =0$. Thus $F(\psi,\phi)=0$. Therefore $\psi \cdot \phi=F(\psi,\phi)^2$.

\item
Suppose that $\pi_{\psi}$ and $\pi_{\phi}$ are unitarily equivalent. Then there is a unitary element $U$ in $\mathfrak{A}$ such that $\phi(A)=\psi(U^*AU)$ for any $A \in \mathfrak{A}$ \cite[Theorem 10.2.6]{kadison1986fundamentals}. Then $x_{\psi} \in S(\pi_{\psi}, \psi)$ and $\pi_{\psi}(U)x_{\psi} \in S(\pi_{\psi}, \phi)$, which imply
\begin{equation}
\label{transition001}
| \langle x_{\psi}, \pi_{\psi}(U)x_{\psi} \rangle | \leq F(\psi,\phi). 
\end{equation}
By Lemma \ref{inner_product} and Equation (\ref{transition001}), 
\[ \begin{split}
&\| \psi - \phi \| =2(1-|\langle x_{\psi}, \pi_{\psi}(U)x_{\psi} \rangle |^2)^{1/2} \geq 2(1-F(\psi,\phi)^2)^{1/2}.
\end{split} \]
On the other hand, $\| \psi - \phi \| \leq 2(1-F(\psi,\phi)^2)^{1/2}$ holds \cite[Section 3]{uhlmann1976transition}. Thus $\| \psi - \phi \| = 2(1-F(\psi,\phi)^2)^{1/2}$, which implies $F(\psi,\phi)^2=1-(1/4)\| \psi - \phi \|^2$. Therefore $\psi \cdot \phi=F(\psi,\phi)^2$.
\end{enumerate}
 \end{proof}


Next we characterize a nonabelian C*-algebra in terms of a fidelity, which is needed for the proof of the main theorem (Theorem \ref{imperfect_theorem}).

\begin{proposition}
\label{nonabelian}
The following conditions are equivalent.
\begin{enumerate}
\item $\mathfrak{A}$ is not abelian.
\item There are pure states $\psi$ and $\phi$ of $\mathfrak{A}$ such that $0 < F(\psi,\phi) < 1$.
\item For any real number $\alpha$ such that $0 \leq \alpha \leq 1$, there are pure states $\psi$ and $\phi$ of $\mathfrak{A}$ such that $F(\psi,\phi) = \alpha$.
\end{enumerate}
\end{proposition}

\begin{proof}
\begin{description}
\item[$1 \Leftrightarrow 2$] $\mathfrak{A}$ is not abelian if and only if there are pure states $\psi$ and $\phi$ of $\mathfrak{A}$ such that $0 < \| \psi - \phi \| < 2$ \cite[Exercise 4.6.26]{kadison1983fundamentals}. Thus $\mathfrak{A}$ is not abelian if and only if there are pure states $\psi$ and $\phi$ of $\mathfrak{A}$ such that $0< F(\psi, \phi) < 1$ by Proposition \ref{transition}.
\item[$3 \Rightarrow 2$] Trivial.
\item[$1 \Rightarrow 3$]
Suppose that $\mathfrak{A}$ is not abelian. Then there are elements $A$ and $B$ in $\mathfrak{A}$ such that $[A,B] \neq 0$ and there is a pure state $\psi$ of $\mathfrak{A}$ such that $\psi([A,B]) \neq 0$, where $[A,B]=AB-BA$ \cite[Theorem 4.3.8]{kadison1983fundamentals}. Let $(\pi_{\psi}, \mathcal{H}_{\psi}, x_{\psi})$ be the GNS representation of $\mathfrak{A}$ induced by $\psi$. The dimension of the Hilbert space $\mathcal{H}_{\psi}$ exceeds one; for, otherwise, $\psi([A,B])=\langle x_{\psi},\pi_{\psi}([A,B])x_{\psi} \rangle=0$ in contradiction with the fact that $\psi([A,B]) \neq 0$. Thus there is a unit vector $y \in \mathcal{H}_{\psi}$ which is orthogonal to $x_{\psi}$ \cite[Proof of Lemma 4]{clifton2003characterizing}.

Let $\alpha$ be a real number such that $0 \leq \alpha \leq 1$, let $\theta$ be a real number such that $ \cos \theta = \alpha$ and let $z=\cos \theta \cdot x_{\psi} + \sin \theta \cdot y$ and let $\phi$ be a state of $\mathfrak{A}$ such that $\phi(A)=\langle z, \pi_{\psi}(A_1) z \rangle$ for any $A \in \mathfrak{A}$. Then $\phi$ is a pure state of $\mathfrak{A}$ \cite[Corollary 10.2.5]{kadison1986fundamentals}. By Lemma \ref{inner_product}, $\| \psi - \phi \|=2(1-| \langle x_{\psi},z \rangle |^2)^{1/2}=2 \sin \theta$. By Proposition \ref{transition}, $F(\psi,\phi)=\cos \theta = \alpha$.
\end{description}
 \end{proof}


\section{No-cloning theorem}
\label{imperfect_cloning}
In this section we examine no-cloning theorem in the above-mentioned C*-algebraic framework. Before defining a perfect cloning operation, we must define a completely positive map, which gives the most general dynamical evolution in a C*-algebraic framework. Recall that a linear mapping $T$ of $\mathfrak{A}$ is positive just in case $A \geq 0$ entails $T(A) \geq 0$. $T$ can be extended to a linear map $T_n$ of $M_n(\mathfrak{A})$ by
\[ T_n 
\begin{pmatrix}
A_{11} & \dots & A_{1n} \\
\cdot & \cdot & \cdot \\
A_{n1} & \dots & A_{nn}
\end{pmatrix}
=
\begin{pmatrix}
T(A_{11}) & \dots & T(A_{1n}) \\
\cdot & \cdot & \cdot \\
T(A_{n1}) & \dots & T(A_{nn})
\end{pmatrix},
\]
where $M_n(\mathfrak{A})$ is the set of $n$ by $n$ matrices with entries which are elements from the C*-algebra $\mathfrak{A}$.
If $T_n$ is positive, $T$ is said to be $n$-positive. If $T$ is $n$-positive for any $n \in \mathbb{N}$, $T$ is said to be completely positive. A positive map $T$ satisfying $T(I)=I$ is called a unital positive map.

If $T$ is a unital $2$-positive map of $\mathfrak{A}$, then $T(A)^*T(A) \leq T(A^*A)$ for any $A \in \mathfrak{A}$ \cite[Proposition 3.3]{paulsen2002completely}. Thus for any state $\psi$ and $\phi$ of $\mathfrak{A}$ and any unital $2$-positive map $T$ of $\mathfrak{A}$,
\begin{equation}
\label{monotone}
F(\psi, \phi) \leq F(T^*\psi, T^*\phi)
\end{equation}
by \cite[Theorem 4.2]{uhlmann1985transition}.




If $T$ is a unital completely positive map of $\mathfrak{A}$ and $\psi$ is a state of $\mathfrak{A}$, then the mapping $T^*$ of the set of all states of $\mathfrak{A}$ can be defined by $(T^*\psi)(A)=\psi(T(A))$ for any state $\psi$ of $\mathfrak{A}$ and any $A \in \mathfrak{A}$. $T$ captures the dynamic change which occurs as the result of some physical process. $\psi$ is the initial state before the process, and $T^* \psi$ is the final state after the process occurs. 

A universal perfect cloning operations is defined as follows (cf. \cite[p.1578]{clifton2003characterizing}).

\begin{definition}
\label{cloning_definition}
Let $T$ be a unital completely positive map of $\mathfrak{A}_1 \vee \mathfrak{A}_2$ and let $\sigma_2$ be a state of $\mathfrak{A}_2$. We say that $T$ is a universal perfect cloning operation just in case that $T^*(\psi_1 \otimes \sigma_2)=\psi_1 \otimes  \psi_2$ for any pure state $\psi_1$ of $\mathfrak{A}_1$, where $\psi_2$ is a state of $\mathfrak{A}_2$ which is isomorphic to $\psi_1$.
\end{definition}

The perfect cloning operation in Definition \ref{cloning_definition} takes $\psi_1$ as an input, and returns $\psi_2$ as an output. Since $\psi_2$ is isomorphic to $\psi_1$, $\psi_2$ is a perfect clone of $\psi_1$.
Clifton, Bub and Halvorson \cite[Theorem 2 and Theorem 3]{clifton2003characterizing} showed the following theorem.

\begin{theorem}[Clifton-Bub-Halvorson]
\label{CBH}
The following conditions are equivalent.
\begin{enumerate}
\item There is a universal perfect cloning operation of $\mathfrak{A}_1 \vee \mathfrak{A}_2$.
\item $\mathfrak{A}_1$ is abelian.
\end{enumerate}
\end{theorem}

Since any individual system's algebra of observables is not abelian in algebraic quantum theory, there is no universal perfect cloning operation in algebraic quantum theory. Next we examine whether an imperfect cloning operation exists or not. In order to tackle this problem, we define an imperfect cloning operation which tolerates a finite loss $\epsilon$ of fidelity in the cloned state.

\begin{definition}
\label{imperfect_definition}
Let $T$ be a unital completely positive map of $\mathfrak{A}_1 \vee \mathfrak{A}_2$, let $\sigma_2$ be a state of $\mathfrak{A}_2$, and let $\epsilon$ be a real number such that $0 \leq \epsilon \leq 1$. We say that $T$ is a universal $\epsilon$-imperfect cloning operation just in case that for any pure state $\psi_1$ of $\mathfrak{A}_1$ there is a pure state $\bar{\psi}_2$ of $\mathfrak{A}_2$ such that $T^*(\psi_1 \otimes \sigma_2)=\psi_1 \otimes  \bar{\psi}_2$ and $F(\psi_2, \bar{\psi}_2) \geq 1 - \epsilon$, where $\psi_2$ is a state of $\mathfrak{A}_2$ which is isomorphic to $\psi_1$.

\end{definition}

The universal $\epsilon$-imperfect cloning operation in Definition \ref{imperfect_definition} takes $\psi_1$ as an input, and returns $\bar{\psi}_2$ as an output. The original state $\psi_1$ and the cloned state $\bar{\psi}_2$ are not entangled, and the loss of fidelity is less than $\epsilon$. If $\epsilon$ is $0$, then $\bar{\psi}_2=\psi_2$, that is, $\bar{\psi}_2$ is a perfect clone of $\psi_1$. Thus a universal $0$-imperfect cloning operation is equal to a universal perfect cloning operation.


In the following theorem, we examine whether a universal $\epsilon$-imperfect cloning operation which is `good' according to fidelity exists or not in algebraic quantum theory. The term `good' means that the loss of fidelity is less than $1/4$.





\begin{theorem}
\label{imperfect_theorem}
Let $\epsilon$ be a real number such that $0 \leq \epsilon < 1/4$. If there is a universal $\epsilon$-imperfect cloning operation of $\mathfrak{A}_1 \vee \mathfrak{A}_2$, then $\mathfrak{A}_1$ is abelian.
\end{theorem}

Using Theorem \ref{CBH} and Theorem \ref{imperfect_theorem}, we can get the following corollary.

\begin{corollary}
\label{summary}
Let $\epsilon$ be a real number such that $0 \leq \epsilon < 1/4$. The following conditions are equivalent.
\begin{enumerate}
\item There is a universal perfect cloning operation of $\mathfrak{A}_1 \vee \mathfrak{A}_2$.
\item There is a universal $\epsilon$-imperfect cloning operation of $\mathfrak{A}_1 \vee \mathfrak{A}_2$.
\item $\mathfrak{A}_1$ is abelian.
\end{enumerate}
\end{corollary}

Generally any individual quantum system's algebra of observables is not abelian. So the universal $\epsilon$-imperfect cloning operations do not exist in algebraic quantum theory in the case where $0 \leq \epsilon < 1/4$.

For the proof of Theorem \ref{imperfect_theorem}, we will need to invoke technical lemmas. 

\begin{lemma}
\label{Bures_pure}
Let $\psi$ be a state of $\mathfrak{A}$, let $\phi$ be a pure state of $\mathfrak{A}$, and let $\pi$ be a representation of $\mathfrak{A}$ on a Hilbert space $\mathcal{H}$ such that $S(\pi,\psi)$ and $S(\pi,\phi)$ are not empty. Then $F(\psi,\phi)=F_{\pi}(\psi, \phi)$.
\end{lemma}

\begin{proof}
Let $\pi'$ be a representation of $\mathfrak{A}$ on a Hilbert space $\mathcal{H}'$ such that $S(\pi',\psi)$ and $S(\pi',\phi)$ are not empty. We will show that $F_{\pi}(\psi,\phi)=F_{\pi'}(\psi,\phi)$.

Let $x$, $y$, $x'$ and $y'$ be vectors in $S(\pi,\psi)$, $S(\pi,\phi)$, $S(\pi',\psi)$ and $S(\pi',\phi)$, respectively, and let $E \in \pi(\mathfrak{A})'$ and $E' \in \pi'(\mathfrak{A})'$ be projections whose ranges are subspaces generated by $\{ \pi(\mathfrak{A})x \}$ and $\{ \pi'(\mathfrak{A})x' \}$, respectively. Define representations $\pi_{E}$ and $\pi'_{E'}$ of $\mathfrak{A}$ as $\pi_E(A)=E\pi(A)E$ and $\pi'_{E'}=E'\pi'(A)E'$ for any $A \in \mathfrak{A}$, respectively. Then $(\pi_E, E\mathcal{H}, x)$ and $(\pi'_{E'}, E'\mathcal{H}', x')$ are GNS representations induced by $\psi$. Thus there is a unitary operator $U$ from $E\mathcal{H}$ onto $E'\mathcal{H}'$ such that $Ux=x'$ and $U^*\pi'_{E'}(A)U=\pi_E(A)$ for any $A \in \mathfrak{A}$ \cite[Proposition 4.5.3]{kadison1983fundamentals}. Therefore $Ux \in S(\pi', \psi)$.

There are vectors $y_1 \in E\mathcal{H}$ and $y_2 \in (I-E)\mathcal{H}$ such that $y=y_1+y_2$. Then
\[ \phi(A)=\langle y, \pi(A) y \rangle = \langle y_1, \pi(A) y_1 \rangle + \langle y_2, \pi(A) y_2 \rangle \]
for any $A \in \mathfrak{A}$. Since $\phi$ is a pure state of $\mathfrak{A}$, $y_1=0$ or $y_2=0$. Thus $y \in E\mathcal{H}$ or $y \in (I-E)\mathcal{H}$.

Suppose that $y \in E\mathcal{H}$. Then
\[ \langle Uy,\pi'(A)Uy \rangle = \langle y, U^*\pi'_{E'}(A)Uy \rangle = \langle y,\pi_E(A) y \rangle = \langle y,\pi(A)y \rangle = \phi(A) \]
for any $A \in \mathfrak{A}$. Thus $Uy \in S(\pi', \phi)$. Therefore $| \langle x,y \rangle | = | \langle Ux,Uy \rangle | \leq F_{\pi'}(\psi, \phi)$.

Suppose that $y \in (I-E)\mathcal{H}$. Then $| \langle x,y \rangle | = 0 \leq F_{\pi'}(\psi, \phi)$.

Taking the supremum over $x \in S(\pi, \psi)$ and $y \in S(\pi, \phi)$, $F_{\pi}(\psi,\phi) \leq F_{\pi'}(\psi,\phi)$. Similarly $F_{\pi}(\psi,\phi) \geq F_{\pi'}(\psi,\phi)$. Therefore $F_{\pi}(\psi,\phi) = F_{\pi'}(\psi,\phi)$. Since $F(\psi, \phi)=\sup_{\pi}F_{\pi}(\psi,\phi)$, $F(\psi,\phi)=F_{\pi}(\psi,\phi)$.
\end{proof}

\begin{lemma}
\label{sakai}
Let $\psi$ and $\phi$ be states of $\mathfrak{A}$, let $\pi$ be Bures representation and let $x$ be in $S(\pi, \psi)$. If $\phi \leq a \psi$ for some $a>0$, there exists a positive operator $Z \in \pi(\mathfrak{A})''$ such that $\phi(A)=\langle Zx, \pi(A)Zx \rangle$ for any $A \in \mathfrak{A}$.
\end{lemma}

\begin{proof}
Since $\pi$ is Bures representation, there is a vector $y$ in $S(\pi, \phi)$. For any $A \in \mathfrak{A}$, $\langle y, \pi(A^*A) y \rangle \leq a \langle x, \pi(A^*A)x \rangle$. So for any $X \in \pi(\mathfrak{A})''$, $\langle y, X^*X y \rangle \leq a \langle x, X^*X x \rangle$. 

By Sakai-Radon-Nykod\'ym theorem \cite[Theorem 7.3.6]{kadison1986fundamentals}, there is a positive operator $Z \in \pi(\mathfrak{A})''$ such that $\langle y, X y \rangle=\langle Zx, X Zx \rangle$ for any $X \in \pi(\mathfrak{A})''$. Thus $\phi(A)=\langle y, \pi(A)y \rangle=\langle Zx, \pi(A)Zx \rangle$ for any $A \in \mathfrak{A}$.
\end{proof}

\begin{lemma}
\label{promislow}
Let $\psi$ and $\phi$ be states of $\mathfrak{A}$, let $\pi$ be a representation of $\mathfrak{A}$, let $x$ be in $S(\pi, \psi)$ and let $Z$ be a positive operator in $\pi(\mathfrak{A})''$ such that $\phi(A)=\langle Zx, \pi(A)Zx \rangle$ for any $A \in \mathfrak{A}$. Then $F_{\pi}(\psi,\phi)=\langle x, Zx \rangle$.
\end{lemma}

\begin{proof}
Since $x \in S(\pi, \psi)$ and $Zx \in S(\pi, \phi)$, $\langle x, Zx \rangle \leq F_{\pi}(\psi, \phi)$. We will show that $F_{\pi}(\psi, \phi) \leq \langle x, Zx \rangle$.

Let $x' \in S(\pi, \psi)$ and $y' \in S(\pi, \phi)$. Since $x$ and $x'$ induce the same state relative to $\pi$, there is a partial isometry $U$ in $\pi(\mathfrak{A})'$ such that $x'=Ux$. Similarly there is a partial isometry $V$ in $\pi(\mathfrak{A})'$ such that $y'=VZx$. Therefore
$| \langle x',y' \rangle | = |\langle Ux,VZx \rangle | =| \langle V^*UZ^{1/2}x, Z^{1/2}x \rangle | \leq \| V^*UZ^{1/2}x \| \| Z^{1/2}x \| \leq \| Z^{1/2}x \|^2 =\langle x, Zx \rangle$. Taking the supremum over $x' \in S(\pi, \psi)$ and $y' \in S(\pi, \phi)$, $F_{\pi}(\psi, \phi) \leq \langle x, Zx \rangle$. Therefore $F_{\pi}(\psi,\phi)=\langle x, Zx \rangle$.
\end{proof}

The following lemma can be shown in a similar way to the proof of \cite[Lemma 2.3]{promislow1971kakutani}.

\begin{lemma}
\label{product_lemma}
Let $\psi_i$ and $\phi_i$ be states of $\mathfrak{A}_i$ and let $\phi'_j=(1-a)\phi_j+a\psi_j$ for $j=1,2$, where a is a real number such that $0<a<1$. Then
\begin{itemize}
\item $| F(\psi_1 \otimes \psi_2, \phi_1 \otimes \phi_2)-F(\psi_1 \otimes \psi_2, \phi'_1 \otimes \phi'_2)|<20a^{1/2}$,
\item $|F(\psi_i, \phi_i)-F(\psi_i,\phi'_i)|<10a^{1/2}$.
\end{itemize}
\end{lemma}

The following lemma can be found in \cite[Proposition IV.4.13]{takesaki2002theory}. 

\begin{lemma}
\label{minimal}
Let $\pi_1$ and $\pi_2$ be representations of $\mathfrak{A}_1$ and $\mathfrak{A}_2$, respectively. Then $Z_1 \otimes Z_2 \in \{ (\pi_1 \otimes \pi_2)(\mathfrak{A}_1 \otimes \mathfrak{A}_2) \}''$ for any $Z_1 \in \pi_1(\mathfrak{A}_1)''$ and $Z_2 \in \pi_2(\mathfrak{A}_2)''$.
\end{lemma}

\begin{lemma}
\label{product}
Let $\psi_j$ and $\phi_j$ be states of $\mathfrak{A}_j$, and let $\pi_j$ be Bures representation of $\mathfrak{A}_j$ for $j=1,2$. Then $F_{\pi_1 \otimes \pi_2}(\psi_1 \otimes \psi_2, \phi_1 \otimes \phi_2) =F(\psi_1, \phi_1)F(\psi_2, \phi_2)$.
\end{lemma}

\begin{proof}
For any $0 < a < 1$, let $\phi_j'$ be defined as $\phi_j' := (1-a)\phi_j+a \psi_j$. Then $\psi_j(A^*A) \leq (1/a)\phi_j'(A^*A)$ for all $A_j \in \mathfrak{A}_j$. Let $x_j \in S(\pi_j,\phi_j')$. By Lemma \ref{sakai}, there exists a positive operator $Z_j \in \pi_j(\mathfrak{A}_j)''$ such that $\psi_j(A_j)=\langle Z_j x_j, \pi_j(A_j)Z_j x_j \rangle$ for any $A_j \in \mathfrak{A}_j$. Then $x_1 \otimes x_2 \in S(\pi_1 \otimes \pi_2, \phi_1' \otimes \phi_2')$ and $(Z_1 \otimes Z_2)(x_1 \otimes x_2) \in S(\pi_1 \otimes \pi_2, \psi_1 \otimes \psi_2)$. By Lemma \ref{minimal}, $Z_1 \otimes Z_2 \in \{ (\pi_1 \otimes \pi_2)(\mathfrak{A}_1 \otimes \mathfrak{A}_2) \}''$. Thus $F_{\pi_1 \otimes \pi_2}(\psi_1 \otimes \psi_2, \phi_1' \otimes \phi_2') = \langle x_1 \otimes x_2, (Z_1 \otimes Z_2)(x_1 \otimes x_2) \rangle = \langle x_1, Z_1 x_1 \rangle \langle x_2, Z_2 x_2 \rangle$ and $F(\psi_j, \phi_j')=\langle x_j, Z_jx_j \rangle$ by Lemma \ref{promislow}. Therefore
\begin{equation}
\label{e1}
F_{\pi_1 \otimes \pi_2}(\psi_1 \otimes \psi_2, \phi_1' \otimes \phi_2')=F(\psi_1, \phi_1')F(\psi_2, \phi_2')
\end{equation}

By Lemma \ref{product_lemma}, $|F_{\pi_1 \otimes \pi_2}(\psi_1 \otimes \psi_2, \phi_1 \otimes \phi_2) - F_{\pi_1 \otimes \pi_2}(\psi_1 \otimes \psi_2, \phi_1' \otimes \phi_2')| \leq 20 a^{1/2}$ and $|F(\psi_j, \phi_j')-F(\psi_j,\phi_j)|<10a^{1/2}$. Using Equation (\ref{e1}),
\[ \begin{split}
&| F_{\pi_1 \otimes \pi_2}(\psi_1 \otimes \psi_2, \phi_1 \otimes \phi_2)-F(\psi_1,\phi_1)F(\psi_2,\phi_2)| \\
&=| F_{\pi_1 \otimes \pi_2}(\psi_1 \otimes \psi_2, \phi_1 \otimes \phi_2)- F_{\pi_1 \otimes \pi_2}(\psi_1 \otimes \psi_2, \phi_1' \otimes \phi_2') \\
& \ \ \ \ \ \ \ +F(\psi_1, \phi_1')F(\psi_2, \phi_2')-F(\psi_1,\phi_1)F(\psi_2,\phi_2)| \\
&\leq | F_{\pi_1 \otimes \pi_2}(\psi_1 \otimes \psi_2, \phi_1 \otimes \phi_2)- F_{\pi_1 \otimes \pi_2}(\psi_1 \otimes \psi_2, \phi_1' \otimes \phi_2') | \\
& \ \ \ \ \ \ \ + |F(\psi_1, \phi_1')F(\psi_2, \phi_2')-F(\psi_1,\phi_1')F(\psi_2,\phi_2)| \\
& \ \ \ \ \ \ \ +|F(\psi_1, \phi_1')F(\psi_2, \phi_2)-F(\psi_1,\phi_1)F(\psi_2,\phi_2)| \\
&\leq 20a^{1/2}+10a^{1/2}+10a^{1/2}=40a^{1/2}.
\end{split} \]
Since $a$ is an arbitrary real number such that $0 < a < 1$, $F_{\pi_1 \otimes \pi_2}(\psi_1 \otimes \psi_1, \phi_1 \otimes \phi_2)=F(\psi_1, \phi_1)F(\psi_2, \phi_2)$. 
\end{proof}

\begin{proof}[Proof of Theorem \ref{imperfect_theorem}]

Let $T$ be a universal $\epsilon$-imperfect cloning operation of $\mathfrak{A}_1 \vee \mathfrak{A}_2$. The proof proceeds by contradiction.

Suppose that $\mathfrak{A}_1$ is not abelian. By Proposition \ref{nonabelian}, there exist pure states $\psi_1$ and $\phi_1$ such that $0<F(\psi_1, \phi_1)<1 - 4\epsilon$ since $\epsilon < 1/4$. By Equation (\ref{e0.01}),
\begin{equation}
\label{theorem0.1}
0<F(\psi_2, \phi_2)<1 - 4\epsilon,
\end{equation}
where $\psi_2$ and $\phi_2$ are isomorphic to $\psi_1$ and $\phi_1$, respectively. By the definition of $T$, there are pure states $\bar{\psi_2}$ and $\bar{\phi_2}$ of $\mathfrak{A}_2$ and a state $\sigma_2$ of $\mathfrak{A}_2$ such that $T^*(\psi_1 \otimes \sigma_2)=\psi_1 \otimes \bar{\psi_2}$, $T^*(\phi_1 \otimes \sigma_2)=\phi_1 \otimes \bar{\phi_2}$, $F(\psi_2, \bar{\psi}_2) \geq 1- \epsilon$ and $F(\phi_2, \bar{\phi}_2) \geq 1 - \epsilon$. In this case,
\begin{equation}
\label{theorem0.3}
D(\psi_2, \bar{\psi}_2) \leq \sqrt{2 \epsilon}, \ \ \ D(\phi_2, \bar{\phi}_2) \leq \sqrt{2 \epsilon}
\end{equation}
by Equation (\ref{fidelity_distance}).

Since $\psi_1 \otimes \bar{\psi}_2$ is a pure state of $\mathfrak{A}_1 \otimes \mathfrak{A}_2$,
\begin{equation}
\label{theorem_01}
F(\psi_1 \otimes \bar{\psi_2}, \phi_1 \otimes \bar{\phi_2}) =F_{\pi_1 \otimes \pi_2}(\psi_1 \otimes \bar{\psi_2}, \phi_1 \otimes \bar{\phi_2})=F(\psi_1,\phi_1)F(\bar{\psi_2},\bar{\phi_2})
\end{equation}
\begin{equation}
\label{theorem_02}
F(\psi_1, \phi_1)=F(\psi_1, \phi_1)F(\sigma_2,\sigma_2)=F_{\pi_1 \otimes \pi_2}(\psi_1 \otimes \sigma_2, \phi_2 \otimes \sigma_2),
\end{equation}
by Lemma \ref{Bures_pure} and Lemma \ref{product}.

Using Equations (\ref{monotone}), (\ref{theorem_01}) and (\ref{theorem_02}),
\[ \begin{split}
F(\psi_1, \phi_1)&=F_{\pi_1 \otimes \pi_2}(\psi_1 \otimes \sigma_2, \phi_1 \otimes \sigma_2) \\
&\leq F(\psi_1 \otimes \sigma_2, \phi_1 \otimes \sigma_2) \\
&\leq F(T^*(\psi_1 \otimes \sigma_2), T^*(\phi_1 \otimes \sigma_2)) \\
&=F(\psi_1 \otimes \bar{\psi_2}, \phi_1 \otimes \bar{\phi_2}) \\
&=F(\psi_1,\phi_1)F(\bar{\psi_2},\bar{\phi_2}).
\end{split} \]
Since $F(\psi_1, \phi_1) \neq 0$, $F(\bar{\psi}_2, \bar{\phi}_2)=1$. It implies $\bar{\psi}_2=\bar{\phi}_2$. By Equation (\ref{triangle}),
\begin{equation}
\label{theorem03}
D(\psi_2,\phi_2) \leq D(\psi_2, \bar{\psi}_2) + D(\bar{\psi}_2, \phi_2)=D(\psi_2, \bar{\psi}_2) + D(\bar{\phi}_2, \phi_2). 
\end{equation}
Inequalities (\ref{theorem0.3}) and (\ref{theorem03}) imply that $D(\psi_2, \phi_2) \leq 2\sqrt{2 \epsilon}$. By Equation (\ref{fidelity_distance}), \begin{equation}
\label{theorem0.2}
F(\psi_2,\phi_2) \geq 1 - 4\epsilon.
\end{equation}
It contradicts with Inequality (\ref{theorem0.1}). Therefore $\mathfrak{A}_1$ is abelian.
\end{proof}

\section{Summary}

Clifton, Bub and Halvorson \cite{clifton2003characterizing} defined a perfect cloning operation in the C*-algebraic framework, and showed that an individual system's algebra of observables is abelian if and only if there is a universal perfect cloning operation (Theorem \ref{CBH}). 
Thus there is no universal perfect cloning operation in algebraic quantum theory. On the other hand, Bu{\v{z}}ek and Hillery \cite{buvzek1996quantum} showed that there exists a universal imperfect cloning operation which tolerates a finite loss of fidelity in the cloned state in quantum mechanics. In this operation, the original and the cloned states are entangled.

In the present paper, we examined the case where the original and the cloned states are not entangled, and a finite loss of fidelity in the cloned state is tolerated. In Definition \ref{imperfect_definition} we defined a universal $\epsilon$-imperfect cloning operation. This operation takes $\psi_1$ as an input, and returns $\bar{\psi}_2$ as an output. The original state $\psi_1$ and the cloned state $\bar{\psi}_2$ are not entangled, and the loss of fidelity is less than or equal to $\epsilon$. In Corollary \ref{summary} it is shown that $\mathfrak{A}_1$ is abelian if and only if there is a universal $\epsilon$-imperfect cloning operation in the case where the loss of fidelity is less than $1/4$. Therefore in this case no universal $\epsilon$-imperfect cloning operation is possible in algebraic quantum theory.




\section*{Acknowledgements}
The author wishes to thank Izumi Ojima and Yutaka Shikano for helpful comments on an earlier draft.
The author is supported by the JSPS KAKENHI, No.23701009 and the John Templeton Foundation Grant ID 35771.

\end{document}